\documentclass{article}
\usepackage{amsmath,amssymb,amsthm,color,bm,mathrsfs}
\usepackage{hyperref} \hypersetup{colorlinks=true}
\usepackage{color}
\newtheorem{theorem}{Theorem}[section]
\newtheorem{corollary}[theorem]{Corollary}

\newtheorem{lemma}[theorem]{Lemma}
\newtheorem{definition}[theorem]{Definition}

\newtheorem{remark}[theorem]{Remark}

\numberwithin{equation}{section}

\usepackage{citeref}

\def\Z{{\Bbb Z}} \def\R{{\cal R}} \def\I{{\cal I}}
\def\la{\lambda} \def\al{\alpha} \def\be{\beta}
\def\ol{\overline} \def\wh{\widehat}

\begin{document}
\title{
Double Constacyclic Codes over Two Finite Commutative Chain Rings}
\author{Yun Fan\\
{\small School of Mathematics and Statistics}\\
{\small Central China Normal University, Wuhan 430079, China}\\[9pt]
Hualu Liu\\
{\small School of Science}\\
{\small Hubei University of Technology, Wuhan 430068, China}}
\date{} 
\maketitle

\insert\footins{\footnotesize{\it Email address}:
yfan@ccnu.edu.cn (Yun Fan), hwlulu@aliyun.com (Hualu Liu).}

\begin{abstract}
Many kinds of codes which possess two cycle structures over two
special finite commutative chain rings,
such as ${\Bbb Z}_2{\Bbb Z}_4$-additive cyclic codes
and quasi-cyclic codes of fractional index etc., were proved asymptotically good.
In this paper we extend the study in two directions: we consider any two
finite commutative chain rings with a surjective homomorphism
from one to the other, and consider double constacyclic structures.
We construct an extensive kind of double constacyclic codes
over two finite commutative chain rings. And,
developing a probabilistic method suitable for quasi-cyclic codes over fields,
we prove that the double constacyclic codes
over two finite commutative chain rings are asymptotically good.

\end{abstract}

\section{Introduction}\label{introduction}
All rings in this paper are commutative and with identity.
By $R^\times$ we denote the multiplicative group
of units (invertible elements) of a ring $R$.
A ring is called a {\em finite chain ring} if it is finite and
its ideals form a chain with respect to the inclusion relation.
A finite ring $R$ is a chain ring if and only if there is
a nilpotent element $\pi$ of $R$ such that the quotient $\ol R=R/R\pi=:F$ is a field,
which is called the {\em residue field} of $R$; for $a\in R$,
$\ol a\in F$ denotes the residue image of~$a$;
see Remark \ref{chain ring} below for details.
Finite fields (i.e., Galois fields), residue integer rings $\Z_{p^s}$
modulo prime power $p^s$, Galois rings etc. are special kinds of finite chain rings.

A code sequence $C_1,C_2,\cdots$ is said to be {\em asymptotically good}
if the code length of $C_i$ goes to infinity, and both the
rates and the relative minimum distances of $C_i$'s are positively 
bounded from below; see Definition \ref{goodness} below for details.
A class of codes is said to be asymptotically good if there is
an asymptotically good code sequence in the class.   

It is well-known that linear codes over a finite field are asymptotically good,
see \cite{G}, \cite{V}. More precisely, in \cite{P67}
the relative minimum distance of linear codes are proved
to be asymptotically distributed at the so-called GV-bound.
In \cite{BF}, the asymptotic distribution
of the relative minimum distance of general codes
(without any more algebraic structures) is characterized.
On the other hand, for the quasi-abelian codes with index going to infinity,
\cite{FL15} proved that the relative minimum distances of such codes
 are also asymptotically distributed at the GV-bound.

Cyclic codes over finite fields are studied
and applied extensively, e.g. see~\cite{HP}.
However it is still an attractive open question: are cyclic codes over a finite field
asymptotically good? For example, see \cite{MW06}.

Quasi-cyclic codes of index $2$ over a finite fields are asymptotically good,
see \cite{CPW}, \cite{C}, \cite{K}.
Bazzi and Mitter \cite{BM} extended
the result to quasi-abelian codes of index $2$.
Soon after, \cite{MW} proved that
binary double even quasi-cyclic codes of index $2$ are asymptotically good.

Self-dual quasi-cyclic codes (with index going to infinity)
over a finite field are asymptotically good,
see \cite{D}, \cite{LS}.
Based on Artin's primitive element conjecture,
Alahmadi et al.~\cite{AOS} proved that, if $q$ is not a square,
 $q$-ary self-dual quasi-cyclic codes of index $2$ are asymptotically good.
In both \cite{Lin} and \cite{LF}, the asymptotic goodness
of any $q$-ary self-dual quasi-cyclic codes of index $2$
is obtained. 

In \cite{FL16} we introduced the quasi-cyclic codes of fractional index
over finite fields, and proved the asymptotically goodness of them.
Mi and Cao \cite{MC} extended the result.
 Gao et al. \cite{GFF} studied
the algebraic structure of quasi-cyclic codes of index $1\frac{1}{2}$.
Aydin and Halilovi\'c \cite{AH} introduced the so-called multi-twisted codes
over finite fields, which are much more extensive than
the quasi-cyclic codes of fractional index.

Hammons et al. \cite{HKCSS} initiated the study on coding over finite rings.
Research on codes over finite chain rings is developed very much,
e.g., \cite{DL}, \cite{FNKS}, \cite{NS}.

On the other hand,
$\Z_2\Z_4$-additive codes appeared in Delsarte \cite{Del}.
It is extended to $\Z_{p^r}\Z_{p^s}$-additive codes in
Aydogd and Siap \cite{AS}.
Abualrub et al.~\cite{ASA} introduced $\Z_2\Z_4$-additive cyclic codes.
Borges et al. generalized the results on $\Z_2\Z_4$-additive cyclic codes
 to $\Z_{p^r}\Z_{p^s}$-additive cyclic codes in \cite{BFT18}.
Gao et al.~\cite{GSWF} investigated the double cyclic codes
(quasi-cyclic codes of index $2$) over $\Z_4$.
In \cite{Liu} and \cite{FL19}, we found that $\Z_2\Z_4$-additive cyclic codes
 are asymptotically good. It is extended to that
 $\Z_{p^r}\Z_{p^s}$-additive cyclic codes are asymptotically good,
see Yao et al.~\cite{YZK}.
On the  other hand, Gao et al. \cite{GH} show that
$\Z_4$-double cyclic codes are asymptotically good.

Generalizing $\Z_2\Z_4$ codes, $\Z_{p^r}\Z_{p^s}$ codes etc.,
Borges et al.~\cite{BFT16} introduced a general type of codes
over two finite chain rings $R_1$ and $R_2$
with an epimorphism (i.e., surjective homomorphism) $R_1\to R_2$.
They called these codes by $R_1R_2$-linear codes,
and by $R_1R_2$-linear cyclic codes if further cyclic structures are afforded;
They investigated the algebraic structures of these codes.

In this paper we introduce a more general type of codes as follows.
Let $R$ and $R'$ be finite chain rings
with an epimorphism $\rho:R\!\to\! R'$,
hence they have the same residue field $\ol R\!=\!\ol{R'}\!=:\!F$.
Let $\la\in R^\times$, and $\la'=\rho(\la)\in R'^\times$;
hence $\ol\la=\ol{\la'}\in F^\times$.
Let $t={\rm ord}_{F^\times}\!(\ol\la)$
be the order of the residue element $\ol\la$ in the unit group $F^\times$.
Assume that $\al,\al'$ are positive integers
such that $\al\equiv \al'~({\rm mod}~t)$ and $\gcd(\al,t)=1$.
As usual, $R[X]$ denotes the polynomial ring over $R$.
For any integer $n>0$ the quotient ring
$R[X]\big/\langle X^{\al n}-\lambda \rangle$
of $R[X]$ modulo the ideal $\langle X^{\al n}-\lambda \rangle$
generated by $X^{\al n}-\lambda$ is an $R[X]$-module.
Through the epimorphism~$\rho$,
the quotient $R'[X]\big/\langle X^{\al'\! n}-\lambda'\rangle$
is also an $R[X]$-module.
We define {\em double constacyclic codes over $(R',R)$} as follows
 (please see Definition \ref{R'R-} below for more details):
\begin{itemize}
\item\vskip-7pt
Any $R[X]$-submodule $C$ of
 $\big(R'[X]\big/\langle X^{\al' n}-\lambda' \rangle\big) \times
  \big( R[X]\big/\langle X^{\al n}-\lambda \rangle\big) $
 is said to be an {\em $(R',R)$-linear $(\la',\la)$-constacyclic codes},
 the pair $(\al'n, \al n)$ is called the {\em cycle length} of $C$,
 the fraction $\al'/\al$ is called the {\em ratio} of $C$.
\end{itemize}
\vskip-7pt
With notation as above, we will prove that such codes are asymptotically good.

\begin{theorem}\label{main}
There are positive integers $n_1,n_2,\cdots$,
and for each $n_i$, $i=1,2,\cdots$, there is an
$(R',R)$-linear $(\la',\la)$--constacyclic code $C_i$ of cycle length $(\al' n_i, \al n_i)$
such that the code sequence $C_1,C_2,\cdots$ is asymptotically good.
\end{theorem}

Obviously, the asymptotic goodness of many kinds
of codes mentioned above are straightforward
consequences of the theorem.
Further, we mention two special interesting cases.

If $R'=R=F$ is a finite field (which is an important case),
then $\la'=\la$ and, following notation of \cite{AH},
we also call the $(F,F)$-linear $(\la,\la)$-constacyclic codes of ratio $\al'/\al$ by
{\em double $\la$-twisted codes of ratio $\al'/\al$} over~$F$.
In particular:

$\bullet$ if $\al'=\al$, then double $\la$-twisted codes of ratio $\al'/\al$ are just quasi-constacyclic codes of index $2$;

$\bullet$ if $\la=1$ (hence $\al,\al'$ are arbitrary positive integers),
then $(F,F)$-linear $(1,1)$-constacyclic codes of ratio $\al'/\al$
are double cyclic codes of ratio $\al'/\al$;
and they are just quasi-cyclic codes of index~$1\frac{1}{\al}$
once $\al'=1$.

\begin{corollary}\label{twisted}
The double $\la$-twisted codes of ratio $\al'/\al$
over a finite field $F$ are asymptotically good.
\end{corollary}

In particular, the following two hold.

\begin{corollary}\label{quasi-constacycic}
The quasi-constacyclic codes of index $2$ over any finite field
are asymptotically good.
\end{corollary}

\begin{corollary}\label{fractional}
For any positive integer $\al$,
the quasi-cyclic codes of index $1\frac{1}{\al}$ over any finite field
are asymptotically good.
\end{corollary}

If $R'=\Z_{p^r}$ and $R=\Z_{p^s}$ with $r\le s$, 
then $(R',R)$-linear $(\la',\la)$-constacyclic codes are just
$(\Z_{p^r},\Z_{p^s})$-additive $(\la',\la)$-constacyclic codes.

\begin{corollary}\label{Z_pZ_p}
The $(\Z_{p^r},\Z_{p^s})$-additive $(\la',\la)$-constacyclic codes of ratio $\al'/\al$
are asymptotically good. In particular,
$\Z_{p^r}\Z_{p^s}$-additive cyclic codes of any ratio $\al'/\al$
are asymptotically good.
\end{corollary}

To prove Theorem \ref{main}, we develop a probabilistic method, which was
proved effective for the quasi-cyclic codes of index $2$ over finite fields,
to a method suitable for the double twisted codes over finite fields.
That is a special case of Theorem \ref{main}. 
Then, using the minimal ideals of $R'$ and $R$,
we make a reduction of the main result over general finite chain rings
to the case over finite fields.
The two skills are another contributions of the paper.

In Section \ref{preliminaries} we describe
the related notation more detailed, and sketch necessary preliminaries.
In Section \ref{double twisted over fields}, we consider the case
that $R'=R=F$ is a finite field, and investigate the
double twisted codes over $F$; by developing a probabilistic method,
we prove the main result Theorem \ref{main} in that case,
i.e., Theorem \ref{c good sequence over F} below.
Finally, Section \ref{R'R c-cyclic over rings} is devoted to
completing a proof of Theorem \ref{main} for general case,
i.e., Theorem \ref{c good sequence over R} bellow,
which is a more precise version of Theorem \ref{main}.

\section{Preliminaries}\label{preliminaries}

In this paper any ring $R$ is commutative and with identity $1_R$ (or $1$ for short).
Subrings and ring homomorphisms are all identity-preserving.
$R^\times$ denotes the multiplicative unit group of $R$.
For any set $S$, $|S|$ denotes the cardinality of~$S$.

\begin{remark}\label{chain ring}\rm
A finite ring $R$ is called a {\em chain ring}
if all ideals of $R$ form a chain by the inclusion relation.
It is easy to see that a finite ring $R$ is a chain ring if and only if
$R$ has a nilpotent element $\pi$ such that $\ol R=R/R\pi=:F$ is a field,
hence $R\pi=J(R)$ is the radical of $R$; e.g., see \cite[Lemma 2.4]{F21}.
$F=\ol R$ is called the residue field of $R$;
 for $a\in R$, $\ol a\in F$ denotes the residue image of~$a$.

In the following we assume that $R$ is a finite chain ring with radical $J(R)=R\pi$,
$\pi^\ell=0$ but $\pi^{\ell-1}\ne 0$,
the positive integer $\ell$ is called the {\em nilpotency index} of~$\pi$
(and, of the radical $R\pi$). Then the following hold:
\begin{itemize}
\item \vskip-7pt
$R\supsetneq R\pi\supsetneq\cdots\supsetneq R\pi^{\ell-1}\supsetneq R\pi^\ell=0$
are the all ideals of $R$.
\item \vskip-4pt
$F:=R/R\pi={\rm GF}(p^r)$ is a finite field (Galois field) with
$|F|=q=p^r$, where $p$ is a prime, $F$ is called the residue field of $R$.
\item \vskip-4pt
$R\pi^i/R\pi^{i+1}\cong F$ for $i=0,1,\cdots,\ell-1$,
and the cardinality $|R|=q^\ell$.
\end{itemize}
\vskip-7pt
For more details, please see \cite{McD} or \cite{N08}.
\end{remark}

There exist several weight functions ${\rm w}$ on $R$,
e.g., Hamming weigh, homogeneous weight, etc.
And each weight ${\rm w}$ on $R$ is extended in a natural way to
 a weight function on the $R$-module
$R^n=\{a=(a_0,a_1,\cdots,a_{n-1})\,|\,a_i\in R\}$,
denote by ${\rm w}$ again, which induces a distance on $R^n$
in a natural way: ${\rm d}(a,b)={\rm w}(a-b)$, for all $a,b\in R^n$.
Then for any code $C\subseteq R^n$,
the minimum distance ${\rm d}(C)$ is defined as usual.
If $C$ is linear (i.e., $C$ is an $R$-submodule of $R^n$),
then ${\rm d}(C)={\rm w}(C):=\min\{{\rm w}(c)\,|\,0\ne c\in C\}$,
called the minimum weight of~$C$.

For a code $C\subseteq R^n$,
how to measure the code length and the information length of $C$?
If  $R=F$ is a field (i.e., $\ell=1$) and $|F|=q$,
then $n$ is the code length, while $\log_q|C|$ is the information length of $C$.
Because of the Grey map, for $\Z_2^\alpha\Z_4^\beta$-codes,
an element of $\Z_4$ maybe viewed as an element of  $\Z_2\times\Z_2$.
For $C\subseteq \Z_2^\alpha\Z_4^\beta$,
in some literature; e.g., see \cite{ASA, AS, BFT18, FL19, YZK},
the information length of $C$ is defined to be $\log_2|C|$,
and define the rate of $C$ by ${\rm R}(C)=\frac{\log_2|C|}{\al+2\be}$.

Therefore we define the code length of $C\subseteq R^n$ to be $n\ell$,
and the information length of $C$ to be $\log_q|C|$.
Then the relative minimum distance of $C$ is defined by
$\Delta(C)=\frac{{\rm d}(C)}{n\ell}$,
and the rate of $C$ is defined by ${\rm R}(C)=\frac{\log_q|C|}{n\ell}$.

\begin{definition}\label{goodness}\rm
A sequence of codes $C_1,C_2,\cdots$, where $C_i\subseteq R^{n_i}$,
is said to be {\em asymptotically good} if the length $n_i\ell$ goes to infinity
and there is a positive real number $\delta$ such that
${\rm R}(C_i) \ge\delta$ and $\Delta(C_i)\ge\delta$ for all $i=1,2,\cdots$.

A class of codes is said to be {\em asymptotically good} if there is an
asymptotically good sequence $C_1,C_2,\cdots,$ with every $C_i$ inside the class.
\end{definition}

\begin{lemma}\label{indep of weight}
Let $C_1,C_2,\cdots$, where $C_i\subseteq R^{n_i}$,  be a sequence of codes
with the positive integers $n_i$ going to infinity.
If for a weight function ${\rm w}$ on $R$
the sequence $C_1,C_2,\cdots$ is asymptotically good, then
 for any weight function ${\rm w}'$ on~$R$
 the sequence $C_1,C_2,\cdots$ is asymptotically good.
\end{lemma}

\begin{proof}
Let ${\rm d}(a,b)={\rm w}(a-b)$, ${\rm d}'(a,b)={\rm w}'(a-b)$,
for $a,b\in R$.
Denote $\Delta(C_i)=\frac{{\rm d}(C_i)}{n_i\ell}$ and
$\Delta'(C_i)=\frac{{\rm d}'(C_i)}{n_i\ell}$.
Since $R$ is finite, there is a real number $\omega>0$ such that
${\rm w}'(a)\ge \omega\cdot {\rm w}(a)$ for all $a\in R$.
For any $(a_0,a_1,\cdots,a_{n_i-1})\in R^{n_i}$, we have
$$\textstyle
 {\rm w}'(a_0,a_1,\cdots,a_{n_i\!-\!1})=\sum_{j=0}^{n_i\!-\!1} {\rm w}'(a_j)
 \ge \sum_{j=0}^{n_i\!-\!1} \omega\cdot {\rm w}(a_j
 )=\omega\cdot {\rm w}(a_0,a_1,\cdots,a_{n_i\!-\!1}).
$$
And, for all ${\bf a}=(a_1,\cdots, a_{n_i})$,
${\bf b}=(b_1,\cdots,b_{n_i})\in R^{n_i}$,
$$\textstyle
{\rm d}'({\bf a,b})=\sum_{j=0}^{n_i\!-\!1} {\rm w}'(a_j-b_j)
\ge \sum_{j=0}^{n_i\!-\!1}\omega\cdot {\rm w}(a_j-b_j)
 =\omega\cdot {\rm d}({\bf a,b}).
$$
Thus ${\rm d}'(C_i)\ge\omega\cdot {\rm d}(C_i)$, hence
$\Delta'(C_i)\ge\omega\cdot\Delta(C_i)$.
So, if $\Delta(C_i)\ge\delta$ for all $i=1,2,\cdots$,
then $\Delta'(C_i)\ge\omega\delta$ for all $i=1,2,\cdots$.
Take $\delta'=\min\{\delta, \omega\delta\}$. Then
$\Delta'(C_i)\ge\delta'$ for all $i=1,2,\cdots$.
\end{proof}

\begin{remark}\label{rem weight}\rm
Lemma \ref{indep of weight} shows that the asymptotic goodness of
a code sequence (or, of a class of codes) is independent of
the choice of the weight functions, though in Definition \ref{goodness}
we have to specify a weight function.
Similarly to Lemma \ref{indep of weight},
it can be also proved that the asymptotic goodness is independent of
the choice of the rates.
In the following, therefore, by ``${\rm w}$''
we always denote the Hamming weight; and, 
the minimum distance and the rate of a code $C$ are defined
as that before Definition \ref{goodness}.
\end{remark}

Let $\la\in R^\times$. An $R$-submodule $C\subseteq R^n$ is called a
{\em $\la$-constacyclic code} if
\begin{equation}\label{c cyclic}
(c_0,c_1,\cdots,c_{n-1})\in C \implies
(\la c_{n-1},c_0,\cdots,c_{n-2})\in C.
\end{equation}
Any $\la$-constacyclic code $C\subseteq R^n$ is identified with
an ideal of $R[X]\big/\langle X^n-\la\rangle$, and vice versa;
where $R[X]$ denotes the polynomial ring over $R$ and
$\langle X^n-\la\rangle$ denotes the ideal generated by $X^n-\la$.

Assume that $R'$ is also a finite chain ring
with radical $R'\pi'$ of nilpotency index $\ell'$.
If there is an epimorphism $\rho:R\to R'$,
then the residue field $R'/R'\pi'\cong R/R\pi=F$, and $\ell'\le\ell$.
But the converse is not true in general.
As far as we know, once $R$ and $R'$ are both Galois rings, or both $F$-algebras
($F$ is the residue field), then the epimorphism $\rho:R\to R'$ exist if and only if
$R'/R'\pi'\cong R/R\pi$ and $\ell'\le\ell$, see \cite[Remark 2.5]{F21}.
In general case, the equivalence no longer holds.
It is still an open question how to classify the finite chain rings,
see~\cite{CL}, \cite{Hou01}.
Thus, the assumption in the beginning of Section~3 of \cite{BFT16}
 is invalid in general.
However, the results of~\cite{BFT16} are still valid and interesting provided
such an epimorphism exists.
Inspired by \cite{BFT16}, we introduce a kind of codes which is
more extensive than the kind of $R'R$-linear cyclic codes defined in \cite{BFT16}.

\begin{definition}\label{R'R-}\rm
Let $R$, $R'$ be finite chain rings as above. Assume that
\begin{itemize}
\item \vskip-7pt
 there is an epimorphism $\rho$: $R\to R'$,
hence $\ol R\cong\ol{R'}$ and $\ell'\le\ell$;
\\
in the following we identify $\ol R=\ol{R'}:=F$, and let $|F|=q$;
\item \vskip-5pt
$\la\in R^\times$, ${\rm ord}_{F^\times}\!(\ol\la)=t$; and
   $\la'=\rho(\la)\in R'^\times$, hence $\ol{\la'}=\ol\la$;
\item \vskip-5pt
integers $\alpha, \al'>0$, 
$\al'\equiv\al~({\rm mod}~t)$ and $\gcd(\al,t)=1$.
\end{itemize}
\vskip-7pt
The epimorphism $\rho:R\to R'$ induces
an epimorphism $\rho: R[X]\to R'[X]$, $\sum_i a_i X^i \mapsto \sum_i \rho(a_i)X^i$.
So, for any integer~$n>0$, both
$R[X]\big/\langle X^n-\la\rangle$ and $R'[X]\big/\langle X^n-\la'\rangle$
are $R[X]$-modules.
Any $R[X]$-submodule $C$ of the $R[X]$-module
$$
 \big( R'[X]\big/\langle X^{\al'n}-\lambda' \rangle\big) \times
  \big( R[X]\big/\langle X^{\al n}-\lambda \rangle\big)
$$
 is said to be an {\em $(R',R)$-linear $(\la',\la)$-constacyclic codes},
  the pair $(\al' n, \al n)$ is called the {\em cycle length} of $C$,
 the fraction $\al'/\al$ is called the {\em ratio} of $C$.
Note that, by Remark \ref{rem weight}, the code length of $C$ is
$\al' n\ell'+\al n\ell$, hence the relative minimum distance
$\Delta(C)=\frac{{\rm w}(C)}{\al' n\ell'+\al n\ell}$,
 and the rate ${\rm R}(C)=\frac{\log_q|C|}{\al' n\ell'+\al n\ell}$.

If $R'=R=F$ is a finite field, then $\la'=\la$ and, following notation of \cite{AH},
we also call the $(F,F)$-linear $(\la,\la)$-constacyclic codes of ratio $\al'/\al$
by {\em double $\la$-twisted codes of ratio $\al'/\al$} over~$F$.
\end{definition}

Similarly to Eq.\eqref{c cyclic},
any $(R',R)$-linear $(\la',\la)$-constacyclic code $C$ defined as above
is identified with an $R$-submodule $C\subseteq R'^{\al'\!n}\times R^{\al n}$
satisfying that
\begin{equation}\label{ratio cyclic}
\begin{array}{l}
\quad (c'_0,c'_1,\cdots,c'_{\al' n-1},~c_0,c_1,\cdots,c_{\al n-1})\in C\\[5pt]
 \implies   (\la' c'_{\al'n-1}, c'_0,\cdots,c'_{\al' n-2},
   ~\la c_{\al n-1},c_0,\cdots,c_{\al n-2})\in C;
\end{array}
\end{equation}
and vice versa.

As mentioned in Section \ref{introduction}, we'll prove that
$(R',R)$-linear $(\la',\la)$-constacyclic codes of ratio~$\al'/\al$
are asymptotically good,
and the key step is to prove it for the case that $R'=R=F$
is a finite field, i.e., Theorem \ref{c good sequence over F} below.
Thus we need a few preliminaries about the codes over finite fields.

Let $F$ be a finite field with $|F|=q=p^r$ as above in Definition \ref{R'R-}.
Let $I=\{1,\cdots,n\}$ be an index set, and $F^I=F^n$.
For any subset $I'\subseteq I$, $I'=\{i_1,\cdots,i_k\}$, $1\le i_1<\cdots<i_k\le n$,
denote $F^{I'}=\{(a_{i_1},\cdots,a_{i_k})\mid a_{i_j}\in F\}$; hence
we have the projection $\rho_{I'}:F^I\to F^{I'}$,
$\rho_{I'}(a_1,\cdots,a_n)= (a_{i_1},\cdots,a_{i_k})$.

\begin{definition}\label{d balanced}\rm
A code $C\subseteq F^I$ is said to be {\em balanced} if
there are subsets (repetition allowed) $I_1,\cdots, I_m\subseteq I$
and a positive integer $t$ such that
\begin{itemize}
\item[(1)]
for each $I_j$, $1\le j\le m$, the projection $\rho_{I_j}$ induces
a bijection from $C$ onto $F^{I_j}$
(equivalently,  $q^{|I_j|}=|C|=|\rho_{I_j}(C)|$ for $j=1,\cdots,m$);
\item[(2)]
for each $i\in I$, there are exact $t$ indexes $1\le j_1<\cdots<j_t\le m$
such that $i\in I_{j_h}$ for $h=1,\cdots,t$.
\end{itemize}
\end{definition}

The following function $h_q(\delta)$ is called the {\em $q$-ary entropy}:
\begin{equation}
 h_q(\delta)=\delta\log_q(q-1)-\delta\log_q \delta
   - (1-\delta)\log_q(1-\delta), \quad \delta\in[0,1-q^{-1}],
\end{equation}
where $0\log_q0=0$ as a convention.
The function $h_q(x)$ is strictly increasing and concave
in the interval $[0,1-q^{-1}]$ with $h_q(0)=0$ and $h_q(1-q^{-1})=1$.

The following result was proved in \cite{M74},
\cite{P85} and \cite{S86} for the binary case,
and proved in \cite[Corollary 3.4]{FL15} for general case.

\begin{lemma}\label{<=delta}
Assume that $0<\delta<1-\frac{1}{q}$.
Let $C\subseteq F^n$ be a balanced code. Set $k=\log_q|C|$.
Denote $C^{\le\delta}=\big\{c\in C\,\big|\, \frac{{\rm w}(c)}{n}\le\delta\big\}$.
Then $|C^{\le\delta}|\le q^{k h_q(\delta)}$.
\end{lemma}

The above lemma is suitable to study the asymptotic properties of group codes
such as quasi-abelian codes, dihedral codes etc., e.g., \cite{BM, BW, FL15, FL20}.
To apply it to the investigation of the constacyclic codes,
we need the following lemma.

\begin{lemma}\label{c c balanced}
 Let $C$ 
 be a $\la$-constacyclic code over $F$ of length $n$.
 Then $C$ is a balanced code.
\end{lemma}
\begin{proof} Let $\dim C=k$,
 $I=\{0,1,\cdots,n-1\}$.
Let $I_*=\{i_1,\cdots,i_k\}\subseteq I$ with $0\le i_1<\cdots<i_k<n$ such that
$|I_*|=k=\dim\rho_{I_*}(C)$.
Set
$$
\Theta_{\la}=\begin{pmatrix}0&1\\ &\ddots&\ddots\\
 &&0&1\\ \la&&&0 \end{pmatrix}_{n\times n},
~{\rm so}~
{\bf c}\Theta_\la=(\la c_{n-1},c_0,\cdots,c_{n-2})\in C,
~\forall~{\bf c}\in C.
$$
Let $\theta=(0,1,\cdots,n-1)$ be a cycle permutation,
so $\theta^{-1}I_* =\{i_1-1,\,\cdots,\,i_k-1\}$,
where $i_1-1$ should be replaced by $n-1$ if $i_1=0$.
Further, set
$$
 \Theta_{k}^{-1}=\begin{pmatrix}0&&&1\\ 1&\ddots \\
 &\ddots&0\\ &&1&0 \end{pmatrix}_{k\times k}, ~~~
D_k=\begin{pmatrix}1&&&\\ &\ddots \\
 &&1\\ &&&\la^{-1} \end{pmatrix}_{k\times k}.
$$
It is checked directly that:\\
\indent
--- If $i_1>0$, then
$\rho_{\theta^{-1} I_*}({\bf c})=\rho_{I_*}({\bf c}\Theta_\la)\Theta_k^{-1}$;\\
\indent
--- otherwise, $i_1=0$, and
$ \rho_{\theta^{-1} I_*}({\bf c})=
\rho_{I_*}({\bf c}\Theta_\la)\Theta_k^{-1}{D_k}$.\\
In any case,
$\rho_{\theta^{-1}I_*}(C)=\rho_{I_*}(C)\Theta_k^{-1}=F^{\theta^{-1}I_*}$.
Now we take an information set $I_1\subseteq I$ of $C$,
i.e., $|I_1|=k=\dim\rho_{I_1}(C)$.
Let $I_{j+1}=\theta^{-j}I_1$, $j=0,1,\cdots,n-1$.
Then $I_1,\cdots,I_n$ satisfy Definition \ref{d balanced}(1).
The cyclic permutation group $\langle\theta\rangle=\langle\theta^{-1}\rangle$
acts on $I$ transitively (in fact, regularly).
By Lemma \ref{permutation} below,
$I_1,\cdots,I_n$ satisfy Definition \ref{d balanced}(2).
In conclusion, $C$ is a balanced code.
\end{proof}

\begin{lemma}\label{permutation}
Let a finite group $G$ act transitively on a finite set $X$. Let $Y\subseteq X$,
 $x_0\in X$. Then
$|\{g\in G\mid x_0\in gY\}|=\frac{|G|\cdot|Y|}{|X|}$.
\end{lemma}

\begin{proof}
For any $y\in Y$, there is a $g\in G$ such that $x_0=gy$.
And, for $g'\in G$, $g'y=x_0=gy$ if and only if $g^{-1}x_0=y=g'^{-1}x_0$,
if and only if $g'g^{-1}\in G_{x_0}$,
where $G_{x_0}=\{h\in G\,|\,hx_0=x_0\}$. Thus,
$|\{g\in G\,|\,x_0=gy\}|=|G_{x_0}|=\frac{|G|}{|X|}$.
Therefore, $|\{g\in G\mid x_0\in gY\}|=|Y|\cdot\frac{|G|}{|X|}$.
\end{proof}

\section{Double $\la$-twisted codes over finite fields}
\label{double twisted over fields}
\begin{remark}\label{F-notation}\rm
In this section, we always take the following notation.
\begin{itemize}
\item \vskip-5pt
$F$ is a finite field, the cardinality $|F|=q$;
  $\la\in F^\times$, ${\rm ord}_{F^\times}\!(\la)=t$.
\item \vskip-5pt
$\R=F[X]/\langle X^{n}-1\rangle$, $n$ is a positive integer
with $\gcd(n,qt)=1$.
\item  \vskip-5pt
$\al$, $\al'>0$ are integers,
$\al'\equiv\al~({\rm mod}~t)$ and $\gcd(\al, t)=1$;\\
set $\al''=\min\{\al,\al'\}$.
\item  \vskip-5pt
$\R_{\la,\al}=F[X]/\langle X^{\al n}-\la\rangle$,~
$\R_{\la,\al'}=F[X]/\langle X^{\al' n}-\la\rangle$; ~ then
\\
$\R_{\la,\al'}\times \R_{\la,\al}
 =\{(a',a)\mid a'\in \R_{\la,\al'}, a\in\R_{\la,\al}\}$.
\item  \vskip-5pt
 $\delta\in(0,1-q^{-1})$.
\end{itemize}
 \vskip-5pt
The rings ($F$-algebras) $\R$, $\R_{\la,\al}$ and $\R_{\la,\al'}$
can be viewed as $F[X]$-modules.
Hence $\R_{\la,\al'}\times \R_{\la,\al}$ is an $F[X]$-module.
By Definition \ref{R'R-}, any $F[X]$-submodule
$C\subseteq\R_{\la,\al'}\times \R_{\la,\al}$
is called an $(F,F)$-linear $(\la,\la)$-constacyclic code of ratio $\al'/\al$.
As pointed out in Section \ref{introduction}, 
the code $C$ is also called a double $\la$-twisted code over $F$ of ratio  $\al'/\al$.
The code length of $C$ is $\al' n+\al n$, and the information length of $C$
 is just the dimension $\dim C$, see Remark \ref{rem weight}.
\end{remark}

In Subsection 3.1 we relate an ideal of $\R$ to an ideal of $\R_{\la,\al}$
(and of  $\R_{\la,\al'}$).
In Subsection 3.2 we construct and study
a kind of random double $\la$-twisted codes of ratio $\al'/\al$.
In Subsection 3.3 we prove the asymptotic goodness of the
double $\la$-twisted codes.

\subsection{About $\R$, $\R_{\la,\al}$ and $\R_{\la,\al'}$}

Because $\gcd(n,q)=1$, $X^n-1$ is a product of 
pairwise coprime monic irreducible $F$-polynomials 
$\phi_i(X)$ with degree $\deg\phi_i(X)=d_i$ as follows
$$
X^n-1=\phi_0(X)\phi_1(X)\cdots\phi_m(X),
$$
where we appoint that
\begin{align}\label{phi_0}
\phi_0(X)=X-1,\quad
\wh\phi_0(X)=\phi_1(X)\cdots\phi_m(X)=X^{n-1}+\cdots+X+1.
\end{align}
So $X^n-1=\phi_0(X)\wh\phi_0(X)$.
By Chinese Remainder Theorem,
\begin{equation}\label{decom R}
\R=F_0\oplus F_1\oplus\cdots\oplus F_m,
\end{equation}
where $F_i\cong F[X]/\langle \phi_i(X)\rangle$ are finite fields
and $\dim_F F_i=\deg\phi_i(X)=d_i$ for $i=0,1,\cdots,m$.
Of course, $d_0=1$.

\begin{remark}\label{omega(n)>}\rm
Let notation be as above. We denote
\begin{equation}\label{def mu(n)}
\mu(n)=\min\{d_1,\cdots,d_m\}.
\end{equation}
By \cite[Lemma 2.6]{BM} (for binary case) and 
\cite[Lemma 3.6]{FL20} (for general case),
there is a sequence $n_1,n_2,\cdots$ of positive integers $n_i$ coprime
to~$qt$ such that $\lim\limits_{i\to\infty}\frac{\log_q n_i}{\mu(n_i)}=0$.
Thus, in the following we further assume that $\mu(n)>\log_q n$.
\end{remark}

By Eq.\eqref{decom R}, any ideal  ${\cal J}$ of $\R$
is a direct sum of some of $F_0,F_1,\cdots,F_m$,
hence ${\cal J}=\R e_{\!\cal J}$ for an idempotent $e_{\!\cal J}$,
and $b e_{\!\cal J}=b$ for all $b\in{\cal J}$;
so ${\cal J}$ is a ring with identity $e_{\!\cal J}$.
And, for all $f(X)\in F[X]$, we have $f(X) e_{\!\cal J}\in{\cal J}$ and
\begin{equation}\label{over J}
f(X)b(X)=f(X)e_{\!\cal J}b(X),~~~~\forall~b(X)\in{\cal J};
\end{equation}
that is, the $F[X]$-module structure of ${\cal J}$ is reduced to the
${\cal J}$-module structure of ${\cal J}$ itself.


We relate now $\R$ to a part of $\R_{\la,\al}$ (and $\R_{\la,\al'}$).
We start with a remark.
\begin{remark}\label{1/alpha} \rm
In the multiplicative group $F^\times$, the element
$\la$ generates a cyclic group 
$\langle\la\rangle=\{\la^j\,|\, j\in\Z_t\}$; i.e.,
the elements of $\langle \la\rangle$ are 1-1
corresponding to the elements of $\Z_t$.
By assumption, $\al, n\in\Z_t^\times$, so the inverses
$\frac{1}{\al}, \frac{1}{n}\in\Z_t^\times$ exist,
and $\la^{\frac{1}{\al}}$, $\la^{\frac{1}{n}}$,
$\la^{\frac{1}{\al n}}$, etc., make sense.
Note that, since $\al'\equiv\al~({\rm mod}~t)$,
$\frac{1}{\al'}=\frac{1}{\al}$ in $\Z_t^\times$,
hence $\la^{\frac{1}{\al'}}=\la^{\frac{1}{\al}}$.
\end{remark}

Then we have the decomposition:
\begin{align*}
X^{\al n}\!-\!\la&=(X^n)^\al-(\la^{\frac{1}{\al}})^\al
=(X^n-\la^{\frac{1}{\al}})\cdot\psi_{\la,\al}(X),
\end{align*}
where
\begin{equation}\label{psi}
 \psi_{\la,\al}(X)=(X^n)^{\al-1}
 +(X^n)^{\al-2}\la^{\frac{1}{\al}}
 +\cdots+(\la^{\frac{1}{\al}})^{\al-1}.
\end{equation}
Note that, in the special case ``$\al=1$'',
$\psi_{\la,1}=1$.
By Eq.\eqref{phi_0},
\begin{align*}
X^n-\la^{\frac{1}{\al}}
&
=\la^{\frac{1}{\al}}
\big((X/\la^{\frac{1}{\al n}})^n-1\big)
=\la^{\frac{1}{\al}}\phi_0(X/\la^{\frac{1}{\al n}})
 \wh\phi_0(X/\la^{\frac{1}{\al n}}).
\end{align*}
Setting $\psi_{\la,\al}^+(X)
=\phi_0\big(X/\la^{\frac{1}{\al n}}\big)\psi_{\la,\al}(X)$, we get
\begin{equation}\label{psi^+}
X^{\al n}\!-\!\la=
\la^{\frac{1}{\al}}\phi_0(X/\la^{\frac{1}{\al n}})
 \wh\phi_0(X/\la^{\frac{1}{\al n}})\psi_{\la,\al}(X)
=\la^{\frac{1}{\al}}\wh\phi_0(X/\la^{\frac{1}{\al n}})
 \psi_{\la,\al}^+(X).
\end{equation}
In the special case ``$\la=1$ and $\al =1$'',
$\psi_{1,1}^+(X)=\phi_0(X)$ since $\psi_{1,1}=1$.

Replacing $\al$ by $\al'$,
we get $\psi_{\la,\al'}(X)$ and $\psi_{\la,\al'}^+(X)$
similarly to Eq.\eqref{psi} and Eq.\eqref{psi^+}, respectively.

We are concerned with the following ideals
of $\R$, $\R_{\la,\al}$ and $\R_{\la,\al'}$:
\begin{equation}\label{I,I_la...}
\I=\R\phi_0(X),~~~~
\I_{\la,\al}=\R_{\la,\al}\psi_{\la,\al}^+(X), ~~~~
 \I_{\la,\al'}=\R_{\la,\al'}\psi_{\la,\al'}^+(X),
\end{equation}
which are all $F[X]$-modules,
as $\R$, $\R_{\la,\al}$ and $\R_{\la,\al'}$
are $F[X]$-modules.
By Eq.\eqref{decom R} and Eq.\eqref{def mu(n)},
\begin{equation}\label{I=}
\I=F_1\oplus\cdots\oplus F_m,\quad
\mbox{each $F_i$ is a field with $d_i=\dim_F F_i\ge\mu(n)$;}
\end{equation}
and any $F[X]$-submodule of $\I$ is a direct sum of some of
$F_1,\cdots,F_m$. Then we relate $\I$ to $\I_{\la,\al}$ and $\I_{\la,\al'}$
by the following concept.

\begin{remark}\label{sigma-iso}\rm
Assume that $M,M'$ are $F[X]$-modules, and $\sigma: F[X]\to F[X]$ is an
$F$-algebra automorphism
(i.e., $\sigma$ is both an $F$-linear isomorphism and a ring isomorphism).
If a map $\tau: M\to M'$ preserves additions and satisfies that:
$\tau(am)=\sigma(a)\tau(m)$, $\forall$ $a\in F[X]$, $\forall$ $m\in M$,
then we say that $\tau$ is a {\em $\sigma$-$F[X]$-homomorphism}.
Further, if a $\sigma$-$F[X]$-homomorphism $\tau$ is bijective,
then we say that $\tau$ is a {\em $\sigma$-$F[X]$-isomorphism}.
Note that a $\sigma$-$F[X]$-isomorphism $\tau: M\to M'$ preserves
all the submodule structures, including the dimensions of submodules.
\end{remark}

The following is clearly an $F$-algebra automorphism:
\begin{equation}\label{sigma_la}
\sigma_{\la}:~ F[X] \longrightarrow F[X],~~
 f(X)\longmapsto f\big(X/\la^{\frac{1}{\al n}}\big),
\end{equation}
which is defined for both $\al$ and $\al'$ because
 $\frac{1}{\al'}=\frac{1}{\al}$ in $\Z_t^\times$, see Remark \ref{1/alpha}.
In the special case ``$\la=1$'' (i.e., cyclic case),
$\sigma_{1}={\rm id}_{F[X]}$
is the identity automorphism of $F[X]$.

\begin{lemma}\label{tau_la,al}
The following is a well-defined
$\sigma_{\la}$-$F[X]$-isomorphism:
$$
\tau_{\la,\al}: ~~\I\longrightarrow\I_{\la,\al},~~
 f(X)\longmapsto f\big(X/\la^{\frac{1}{\al n}}\big)\psi_{\la,\al}(X).
$$
\end{lemma}

\begin{proof}
For $f(X)\in\I=\R\phi_0(X)$, $f(X)=g(X)\phi_0(X)$ for a $g(X)\in\R$; then
$$
f\big(X/\la^{\frac{1}{\al n}}\big)\psi_{\la,\al}(X)
=g\big(X/\la^{\frac{1}{\al n}}\big)
 \phi_0\big(X/\la^{\frac{1}{\al n}}\big)\psi_{\la,\al}(X)
=g\big(X/\la^{\frac{1}{\al n}}\big)\psi^+_{\la,\al}(X);
$$
so $f\big(X/\la^{\frac{1}{\al n}}\big)\psi_{\la,\al}(X)\in\I_{\la,\al}$.
Next, assume that both $f(X),f'(X)\in F[X]$ represent one and the same element
in $\I$, then $f'(X)=f(X)+g(X)(X^n-1)$ for a $g(X)\in F[X]$, so
\begin{align*}
f'(X/\la^{\frac{1}{\al n}})\psi_{\la,\al}(X)
=f(X/\la^{\frac{1}{\al n}})\psi_{\la,\al}(X)
+g(X/\la^{\frac{1}{\al n}})((X/\la^{\frac{1}{\al n}})^n\!-\!1)\psi_{\la,\al}(X).
\end{align*}
By Eq.\eqref{psi^+}, in $\I_{\la,\al}$ we have
$((X/\la^{\frac{1}{\al n}})^n-1)\psi_{\la,\al}(X)=0$. Thus
\begin{align*}
f'(X/\la^{\frac{1}{\al n}})\psi_{\la,\al}(X)
=f(X/\la^{\frac{1}{\al n}})\psi_{\la,\al}(X),\quad
(\mbox{in $\I_{\la,\al}$.})
\end{align*}
Summarizing the above, we see that the  $\tau_{\la,\al}$ in the lemma
is a well-defined map.
Obviously, $\tau_{\la,\al}$ preserves additions.
For $f(X)\in\I$ and $g(X)\in F[X]$,
\begin{align*}
& \tau_{\la,\al}\big(g(X)f(X)\big)
=g(X/\la^{\frac{1}{\al n}})f(X/\la^{\frac{1}{\al n}})\psi_{\la,\al}(X)\\
&=g(X/\la^{\frac{1}{\al n}})\tau_{\la,\al}\big(f(X)\big)
=\sigma_{\la}\big(g(X)\big)\tau_{\la,\al}\big(f(X)\big).
\end{align*}
Thus, $\tau_{\la,\al}$ is a $\sigma_{\la}$-$F[X]$-homomorphism.
For any $g(X)\psi_{\la,\al}^+(X)\in\I_{\la,\al}$,
$$
 g(X)\psi_{\la,\al}^+(X)=g(X)\phi_{0}\big(X/\la^{\frac{1}{\al n}}\big)\psi_{\la,\al}(X).
$$
Then $g(\la^{\frac{1}{\al n}}X)\phi_{0}(X)\in\I$ and
\begin{align*}
\tau_{\la,\al}\big(g(\la^{\frac{1}{\al n}}X)\phi_{0}(X)\big)
=g(\la^{\frac{1}{\al n}}X/\la^{\frac{1}{\al n}})
\phi_{0}(X/\la^{\frac{1}{\al n}})\psi_{\la,\al}(X)
=g(X)\psi^+_{\la,\al}(X).
\end{align*}
So, $\tau_{\la,\al}$ is surjective. Finally,
by Eq.\eqref{psi^+} again,
\begin{equation}\label{dim I}
\dim_F\I=\dim_F\I_{\la,\al}=n-1.
\end{equation}
Thus, $\tau_{\la,\al}$ is a bijection.
\end{proof}

\begin{remark}\label{tau_la,be}\rm
It is clear that Lemma \ref{tau_la,al} still holds
if we replace $\al$ by $\al'$; i.e., 
the following is a well-defined $\sigma_{\la}$-$F[X]$-isomorphism
(recall that $\frac{1}{\al'}=\frac{1}{\al}$ in $\Z_t^\times$
hence $\la^{\frac{1}{\al n}}=\la^{\frac{1}{\al' n}}$):
$$\tau_{\la,\al'}:~ \I\longrightarrow\I_{\la,\al'},~~
 f(X)\longmapsto f\big(X/\la^{\frac{1}{\al n}}\big)\psi_{\la,\al'}(X).
$$
For $f(X)\in{\cal I}$, to simplify the notation, in the following
 we'll denote the image $\tau_{\la,\al}\big(f(X)\big)$ by
 $f^{\tau_{\la,\al}}(X)\in {\cal I}_{\la,\al}$,
 and denote $\tau_{\la,\al'}\big(f(X)\big)$ by
 $f^{\tau_{\la,\al'}}(X)\in {\cal I}_{\la,\al'}$.
\end{remark}

\subsection{Random double twisted code $C_{a'\!,a}$ over $F$}

Recall that $\I_{\la,\al'}\times\I_{\la,\al}$
is a $2(n-1)$-dimensional $F[X]$-submodule
of $\R_{\la,\al'}\times\R_{\la,\al}$, see Eq.\eqref{dim I}.
In the rest of this section, we view $\I_{\la,\al'}\!\times\I_{\la,\al}$ as a
probability space with equal probability for every sample.

For $(a'(X),a(X))\in \I_{\la,\al'}\!\times\I_{\la,\al}$,
let $C_{a',a} = F[X]\big(a'(X),a(X)\big)$ be the $F[X]$-submodule of
$\R_{\la,\al'}\!\times \R_{\la,\al}$ generated by $(a'(X),a(X))$, i.e.
\begin{align}\label{C_a',a=ga}
 C_{a'\!,a}
  = \big\{\,(g(X)a'(X),\,g(X)a(X))\in \I_{\la,\al'}\!\times\I_{\la,\al}\;\big|\;
  g(X)\in F[X]\,\big\}.
\end{align}
Then $C_{a'\!,a}$ is a random double $\la$-twisted code
of cycle length $(\al' n,\al n)$ over $F$.
By Lemma~\ref{tau_la,al}, Remark~\ref{tau_la,be},
we can take $b'(X),b(X)\in\I$ such that
$b'^{\tau_{\la,\al'}}(X)=a'(X)$ and $b^{\tau_{\la,\al}}(X)=a(X)$.
For any $g(X)\in F[X]$, by Eq.\eqref{sigma_la}
there is an $f(X)\in F[X]$ such that $\sigma_\la\big(f(X)\big)=g(X)$.
We get that
\begin{align*}
& \tau_{\la,\al}\big(f(X)b(X)\big)
 =\sigma_\la\big(f(X)\big)a(X)=g(X)a(X); \\
& \tau_{\la,\al'}\big(f(X)b'(X)\big)
 =\sigma_\la\big(f(X)\big)a'(X)=g(X)a'(X).
\end{align*}
Further, by Eq.\eqref{over J}, we can take the $f(X)$ such that $f(X)\in\I$.
So we get
\begin{equation}\label{C_a',a=fa}
 C_{a'\!,a}=
 \left\{\big(\sigma_\la(f(X))a'(X),\, \sigma_\la(f(X))a(X)\big)
 \,\big|\,  f(X)\in \I\right\}\subseteq\I_{\la,\al'}\!\times\I_{\la,\al}.
\end{equation}
For each $f(X)\in\I$, we have a random code word
$$
 c_{f,a'\!,a}=\big(\sigma_\la(f(X))a'(X),\,\sigma_\la(f(X))a(X)\big)\in C_{a'\!,a}.
$$
Recall that the code length equals $\al' n+\al n$, see Remark \ref{F-notation}.
We get a $0$-$1$ random variable over
the probability space $\I_{\la,\al'}\!\times\I_{\la,\al}$:
$$
Y_f=\begin{cases}
 1,& 0<\frac{{\rm w}(c_{f,a'\!,a})}{\al' n+\al n}\le\delta;\\
 0, & \mbox{otherwise}. \end{cases}
$$
Clearly, $Y_0=0$.
We further define an non-negative integer random variable
$$\textstyle
Y=\sum_{f(X)\in\I}Y_f.
$$
Because of Eq.\eqref{C_a',a=fa}, the variable $Y$
stands for the number of the non-zero random code words whose relative weight
is at most $\delta$. So
\begin{equation}\label{Delta<E}
 \Pr\big(\Delta(C_{a'\!,a})\le\delta\big)=\Pr\big(Y\ge 1\big)
 \le {\rm E}(Y),
\end{equation}
where ${\rm E}(Y)$ denotes the expectation of $Y$, and
the inequality follows by Markov Inequality,
e.g., see \cite[Theorem 3.1]{MU}.

In the rest of this subsection, we estimate the expectation ${\rm E}(Y)$.
For $f(X)\in\I$,  denote
\begin{equation}\label{C_f}
  C_f=\R f(X)=\I f(X)=\{g(X)f(X)\,|\,g(X)\in\I\}, ~~ d_f=\dim_F C_f.
\end{equation}
Then $C_f$ is an ideal (an $F[X]$-submodule) of $\R$ contained in $\I$.

\begin{lemma}\label{E(Y_f)<}
Keep the above notation. Let $\al''=\min\{\al',\al\}$
and $\delta\in(0,1-q^{-1})$ be as in Remark \ref{F-notation}.
Then the expectation
$${\rm E}(Y_f)\le q^{-2d_f+2d_fh_q\big(\frac{\al'+\al}{2\al''}\delta\big)
 +\log_q((\al'+\al)n)}.
$$
\end{lemma}

\begin{proof}
From that
$c_{f,a'\!,a}=\big(\sigma_\la(f(X))a'(X),\,\sigma_\la(f(X))a(X)\big)$,
we see that
$$ \big\{ c_{f,a'\!,a}\,\big|\, a'(X)\in\I_{\la,\al'},\; a(X)\in\I_{\la,\al}\big\}
= \I_{\la,\al'}\sigma_\la(f(X))\times \I_{\la,\al}\sigma_\la(f(X)).
$$
Denote
$M=\I_{\la,\al'}\sigma_\la(f(X))\times \I_{\la,\al}\sigma_\la(f(X))$.
By the notation in Lemma \ref{<=delta},
\begin{equation}\label{E(Y_f)}
{\rm E}(Y_f)=
\Pr\big({\rm E}(Y_f)=1\big)
 =(|M^{\le\delta}|-1)\big/|M|.
\end{equation}
By Lemma \ref{tau_la,al},
we have a $\sigma_{\la}$-$F[X]$-isomorphism
$$
C_f=\I f(X)\cong \I_{\la,\al}\sigma_\la(f(X));
$$
in particular,
$$
 \dim_F(\I_{\la,\al}\sigma_\la(f(X)))=\dim_F(C_f)=d_f.
$$
In the same way, we have a $\sigma_{\la}$-$F[X]$-isomorphism
\begin{equation}\label{|I|=q^d}
 C_f=\I f(X)\cong \I_{\la,\al'}\sigma_\la(f(X))
 \quad\mbox{hence}\quad
 \dim_F(\I_{\la,\al'}\sigma_\la(f(X)))=d_f.
\end{equation}
So,
\begin{equation}\label{|M|=q^2d}
|M|=\big|\I_{\la,\al'}\sigma_\la(f(X))\times \I_{\la,\al}\sigma_\la(f(X))\big|=q^{d_f}q^{d_f}=q^{2d_f}.
\end{equation}
It is easy to see that
\begin{align*}
M^{\le\delta}&=
\big(\I_{\la,\al'}\sigma_\la(f(X))\times
 \I_{\la,\al}\sigma_\la(f(X))\big)^{\le\delta}\\
&\textstyle=\bigcup\limits_{\scriptsize\begin{matrix}w',w\ge 0\\
 w'+w=\lfloor\delta(\al'+\al)n\rfloor\end{matrix}}
\big(\I_{\la,\al'}\sigma_\la(f(X))\big)^{\le\frac{w'}{\al' n}}\times
\big(\I_{\la,\al}\sigma_\la(f(X))\big)^{\le\frac{w}{\al n}}.
\end{align*}
Thus
\begin{align*}
|M^{\le\delta}|\,
\textstyle\le\sum\limits_{\scriptsize\begin{matrix}w',w\ge 0\\
 w'+w=\lfloor\delta(\al'+\al)n\rfloor\end{matrix}}
\big|\big(\I_{\la,\al'}\sigma_\la(f(X))\big)^{\le\frac{w'}{\al'n}}\big|
\cdot
\big|\big(\I_{\la,\al}\sigma_\la(f(X))\big)^{\le\frac{w}{\al n}}\big|.
\end{align*}
By Lemma \ref{<=delta}, Lemma \ref{c c balanced} and Eq,\eqref{|I|=q^d},
$$
\big|\big(\I_{\la,\al'}\sigma_\la(f(X))\big)^{\le\frac{w'}{\al'n}}\big|
\le q^{d_fh_q(\frac{w'}{\al' n})},\quad
\big|\big(\I_{\la,\al}\sigma_\la(f(X))\big)^{\le\frac{w}{\al n}}\big|
\le q^{d_fh_q(\frac{w}{\al n})}.
$$
Thus
$$
|M^{\le\delta}|\,
 \textstyle\le\sum\limits_{\scriptsize\begin{matrix}w',w\ge 0\\
 w'+w=\lfloor\delta(\al'+\al)n\rfloor\end{matrix}}
 q^{d_f\big(h_q(\frac{w'}{\al' n})+h_q(\frac{w}{\al n})\big)}.
$$
Let $\al^*=\max\{\al',\al\}$. Then $\al''\le\al',\al\le\al^*$
and $\al'\al=\al''\al^*$.
Recall that $h_q(x)$ is concave and increasing
in the interval $[0,1\!-\!\frac{1}{q}]$. So
\begin{align*}
&\textstyle h_q(\frac{w'}{\al' n})+h_q(\frac{w}{\al n})
\le 2h_q\big(\frac{\frac{w'}{\al' n}+\frac{w}{\al n}}{2}\big)
=2h_q\big(\frac{\al w'+\al' w}{2\al'\al n}\big)\\[3pt]
&\textstyle \le 2h_q\big(\frac{\al^* w'+\al^* w}{2\al''\al^* n}\big)
 =2h_q\big(\frac{w'+w}{2\al'' n}\big)
 \le 2h_q\big(\frac{\delta(\al'+\al)n}{2\al'' n}\big)
\textstyle =2h_q\big(\frac{\al'+\al}{2\al''}\delta\big).
\end{align*}
Further, the number of the pairs $(w',w)$ satisfying that
$w',w\ge 0$ and $w'+w=\lfloor\delta(\al'+\al)n\rfloor$ is at most
$(\al'+\al)n$. We obtain
\begin{align*}
|M^{\le\delta}|\,
& \textstyle
 \le (\al'+\al)n\cdot q^{2d_fh_q\big(\frac{\al'+\al}{2\al''}\delta\big)}
 =q^{2d_fh_q\big(\frac{\al'+\al}{2\al''}\delta\big)+\log_q((\al'+\al)n)}.
\end{align*}
Combining it with Eq.\eqref{E(Y_f)} and Eq.\eqref{|M|=q^2d}, we get
$$
{\rm E}(Y_f)\le
q^{2d_fh_q\big(\frac{\al'+\al}{2\al''}\delta\big)+\log_q((\al'+\al)n)}\big/q^{2d_f}
=q^{-2d_f+2d_fh_q\big(\frac{\al'+\al}{2\al''}\delta\big)+\log_q((\al'+\al)n)}.
$$
We are done.
\end{proof}

Let $\mu(n)$ be as in Eq.\eqref{def mu(n)}.
For $\mu(n)\le d\le n-1$, set
\begin{equation}\label{e Omega_d}
 \Omega_d=\{C\,|\,\mbox{$C$ is an $F[X]$-submodule of $\I$},\;\dim_F C=d\}.
\end{equation}
For any ideal $C$ of $\R$, let
\begin{equation}\label{C^*}
 C^*=\{c\in C\,|\,\R c=C\} \quad (\mbox{note that $\R c= Cc$ for $c\in C$}).
\end{equation}

\begin{lemma}\label{l Omega_d}
Let notation be as above. Then

{\bf(1)} $|\Omega_d|\le n^{\frac{d}{\mu(n)}}$.

{\bf(2)}
$\I-\{0\}=\bigcup_{d=\mu(n)}^{n-1}
 \bigcup_{C\in\Omega_d}C^*$.
\end{lemma}

\begin{proof} (1).
By Eq.\eqref{I=},
$\I=F_1\oplus\cdots\oplus F_m$ with
each $F_i$ being a field with $d_i=\dim_F F_i\ge\mu(n)$.
Each $C\in\Omega_d$ is a direct sum of some of $F_1,\cdots,F_m$,
and the number of direct summands is at most $d/\mu(n)$. Thus
$|\Omega_d|\le m^{d/\mu(n)}\le n^{d/\mu(n)}$.

(2).~ For any $f\in\I-\{0\}$, $f\in C_f^*$ and $C_f\in\Omega_{d_f}$.
\end{proof}

\begin{lemma}\label{E(Y)}
Let ${\rm E}(Y)$ be as in Eq.\eqref{Delta<E},
 $\al''\!=\min\{\al',\al\}$ and $\delta\in(0,1-q^{-1})$ as in Remark~\ref{F-notation}.
If $\frac{1}{2}-h_q(\frac{\al'+\al}{2\al''}\delta)
  -\frac{\log_q n}{2\mu(n)}>0$, then
$$
 {\rm E}(Y)\le q^{-2\mu(n)\big(\frac{1}{2}-h_q(\frac{\al'+\al}{2\al''}\delta)
   -\frac{3\log_q n}{2\mu(n)}\big)+\log_q(\al'+\al)}.
$$
\end{lemma}

\begin{proof} Note that $Y_0=0$. By the linearity of
the expectation and Lemma \ref{l Omega_d}(2),
\begin{align*}\textstyle
{\rm E}(Y)=\sum\limits_{f(X)\in\I-\{0\}}{\rm E}(Y_f)
=\sum\limits_{d=\mu(n)}^{n-1}\,\sum\limits_{C\in\Omega_d}
\sum\limits_{f(X)\in C^*}{\rm E}(Y_f).
\end{align*}
By Lemma \ref{E(Y_f)<}, Eq.\eqref{e Omega_d} and Lemma \ref{l Omega_d}(1),
\begin{align*}
\textstyle\sum\limits_{C\in\Omega_d}\,
\sum\limits_{f(X)\in C^*}{\rm E}(Y_f)
&\textstyle\le \sum\limits_{C\in\Omega_d}\sum\limits_{f(X)\in C^*}
q^{-2d+2dh_q(\frac{\al'+\al}{2\al''}\delta)
 +\log_q((\al'+\al)n)}\\
&\textstyle\le \sum\limits_{C\in\Omega_d}
 |C|\cdot q^{-2d+2dh_q(\frac{\al'+\al}{2\al''}\delta)
 +\log_q((\al'+\al)n)}\\
&\textstyle = \sum\limits_{C\in\Omega_d}
  q^{-d+2dh_q(\frac{\al'+\al}{2\al''}\delta)
 +\log_q((\al'+\al)n)}\\
&\le n^{\frac{d}{\mu(n)}}q^{-d+2dh_q(\frac{\al'+\al}{2\al''}\delta)
 +\log_q((\al'+\al)n)}\\
&=q^{-d+2dh_q(\frac{\al'+\al}{2\al''}\delta)
 +\log_q((\al'+\al)n)+\frac{d\log_q n}{\mu(n)}}\\
&=q^{-2d\big(\frac{1}{2}-h_q(\frac{\al'+\al}{2\al''}\delta)
   -\frac{\log_q n}{2\mu(n)}\big)+\log_q((\al'+\al)n)}.
\end{align*}
Since
$\frac{1}{2}-h_q(\frac{\al'+\al}{2\al''}\delta)
  -\frac{\log_q n}{2\mu(n)}>0$ and $d\ge\mu(n)$,
\begin{align*}
  {\rm E}(Y)
  \textstyle\le\sum\limits_{d=\mu(n)}^{n-1}
   q^{-2\mu(n)\big(\frac{1}{2}-h_q(\frac{\al'+\al}{2\al''}\delta)
   -\frac{\log_q n}{2\mu(n)}\big)+\log_q((\al'+\al)n)}.
\end{align*}
The number of the indexes from $\mu(n)$ to $n-1$ is less than $n$. So
\begin{align*}
  {\rm E}(Y)
   &\le n\cdot q^{-2\mu(n)\big(\frac{1}{2}-h_q(\frac{\al'+\al}{2\al''}\delta)
   -\frac{\log_q n}{2\mu(n)}\big)+\log_q n+\log_q(\al'+\al)}\\
   &= q^{-2\mu(n)\big(\frac{1}{2}-h_q(\frac{\al'+\al}{2\al''}\delta)
   -\frac{3\log_q n}{2\mu(n)}\big)+\log_q(\al'+\al)}.
\end{align*}
We are done.
\end{proof}

\begin{lemma}\label{R(C)=}
{\bf(1)} $\dim_F C_{a'\!,a}\le n-1$, i.e., 
${\rm R}(C_{a'\!,a})\le\frac{1}{\al'+\al}-\frac{1}{(\al'+\al)n}$.

{\bf(2)} $\Pr\big(\dim_F C_{a'\!,a}= n-1\big)
  \ge\big(\frac{1}{4}\big)^{\frac{1}{\mu(n)}}$.
\end{lemma}

\begin{proof}
(1). By Eq.\eqref{C_a',a=fa}, $|C_{a'\!,a}|\le |\I|=q^{n-1}$.
That is, $\dim_F C_{a'\!,a}\le n-1$.

(2). If $a(X)\in\I_{\la,\al}$ satisfies that $\I_{\la,\al}a(X)=\I_{\la,\al}$,
i.e., $a(X)\in\I_{\la,\al}^*$ in notation of Eq.\eqref{C^*},
then $\dim_F C_{a'\!,a}=n-1$ (by Eq.\eqref{C_a',a=fa} again).
Thus
$$\textstyle
  \Pr\big(\dim_F C_{a'\!,a}=n-1\big)
  \ge \Pr\big(a(X)\in\I_{\la,\al}^*\big)
  =\frac{|\I_{\la,\al}^*|}{|\I_{\la,\al}|}.
$$
By the isomorphism of Lemma \ref{tau_la,al},
$\frac{|\I_{\la,\al}^*|}{|\I_{\la,\al}|}=\frac{|\I^*|}{|\I|}$. Thus
$$\textstyle
  \Pr\big(\dim_F C_{a,b}=n-1\big)
  \ge\frac{|\I^*|}{|\I|}.
$$
By Eq.\eqref{decom R}, $|\I|=q^{d_1+\cdots+d_m}$, $|\I^*|=(q^{d_1}-1)\cdots(q^{d_m}-1)$,
where $d_1+\cdots+d_m=n-1$ and $d_i\ge\mu(n)$ for $i=1,\cdots,m$
(cf. Eq.\eqref{def mu(n)}).
Hence $m\le\frac{n}{\mu(n)}$.
\begin{align*}\textstyle
\frac{|\I^*|}{|\I|}
 &\textstyle = (1-\frac{1}{q^{d_1}})\cdots(1-\frac{1}{q^{d_m}})
 \ge (1-\frac{1}{q^{\mu(n)}})^m\\
 &\textstyle\ge (1-\frac{1}{q^{\mu(n)}})^{\frac{n}{\mu(n)}}
 =(1-\frac{1}{q^{\mu(n)}})^
  {q^{\mu(n)}\frac{n}{q^{\mu(n)}\mu(n)}}.
\end{align*}
Since the sequence $(1-\frac{1}{h})^h$ for $h=2,3,\cdots$ is increasing and
$(1-\frac{1}{2})^2=\frac{1}{4}$, we have
$(1-\frac{1}{q^{\mu(n)}})^{q^{\mu(n)}}\ge \frac{1}{4}$.
By the assumption in Remark \ref{omega(n)>},
$\mu(n)>\log_q n$, hence
$q^{\mu(n)}>q^{\log_q n}=n$.
We obtain that $\frac{|\I^*|}{|\I|}\ge(\frac{1}{4})^{\frac{1}{\mu(n)}}$.
\end{proof}

\subsection{Asymptotic property of the random code $C_{a'\!,a}$ over $F$}

Keep the notation in Remark \ref{F-notation}.
From Remark \ref{omega(n)>}, we can assume that
positive integers $n_1,n_2,\cdots$ satisfy:
\begin{align}\label{n_i...}\textstyle
\gcd(n_i,qt)=1, ~\forall~i=1,2,\cdots,~~\mbox{and}~~
\lim\limits_{i\to\infty}\frac{\log_q n_i}{\mu(n_i)}=0.
\end{align}
Note that the assumption also implies that $\mu(n_i)>\log_q n_i$
(for $i$ large enough) and
$\mu(n_i)\to\infty$.
In Eq.\eqref{C_a',a=fa}, taking $n=n_i$,
we have the random double $\la$-twisted codes $C_{a'\!,a}^{(i)}$
of cycle length $(\al' n_i,\al n_i)$ over $F$. Then
\begin{align}\label{sequence over F}
 C_{a'\!,a}^{(1)},~ C_{a'\!,a}^{(2)},~\cdots,~C_{a'\!,a}^{(i)},~\cdots
\end{align}
is a sequence of random double $\la$-twisted codes over $F$,
and the length $\al' n_i+\al n_i$ of $C_{a'\!,a}^{(i)}$ goes to infinity.

\begin{theorem}\label{good sequence over F}
Let notation be as in Eq.\eqref{n_i...} and Eq.\eqref{sequence over F}.
Assume that $\delta\in(0,1-q^{-1})$ satisfying that
$h_q(\frac{\al'+\al}{2\al''}\delta)<\frac{1}{2}$, where $\al''=\min\{\al',\al\}$.
 Then

{\bf(1)} $\lim\limits_{i\to\infty}
 \Pr\big(\Delta(C_{a'\!,a}^{(i)})>\delta\big)=1$.

{\bf(2)} $\lim\limits_{i\to\infty}
 \Pr\big(\dim_F C_{a'\!,a}^{(i)}=n_i\!-\!1\big)=1$.
\end{theorem}

\begin{proof}
(1).~ By Eq.\eqref{Delta<E} and Lemma~\ref{E(Y)},
\begin{align*}
\lim\limits_{i\to\infty}\Pr\big(\Delta(C_{a'\!,a}^{(i)})\le\delta\big)
\le \lim\limits_{i\to\infty}
 q^{-2\mu(n_i)\big(\frac{1}{2}-h_q(\frac{\al'\!+\al}{2\al''}\delta)
   -\frac{3\log_q n_i}{2\mu(n_i)}\big)+\log_q(\al'+\al)}.
\end{align*}
Note that $\frac{1}{2}-h_q(\frac{\al'\!+\al}{2\al''}\delta)>0$.
By Eq.\eqref{n_i...}, we have
${\lim\limits_{i\to\infty}\frac{\log_q n_i}{\mu(n_i)}=0}$,
which also implies that $\mu(n_i)\to\infty$.
Then, there is a positive real number $\delta_0$ such that
$\frac{1}{2}-h_q(\frac{\al'\!+\al}{2\al''}\delta)
   -\frac{3\log_q n_i}{2\mu(n_i)}>\delta_0$ for large enough $i$.
So,
$\lim\limits_{i\to\infty}\Pr\!\big(\Delta(C_{a'\!,a}^{(i)})\le\delta\big)=0$.

(2).~ By Lemma \ref{R(C)=},
$
\lim\limits_{i\to\infty}\Pr\big(\dim_F C_{a'\!,a}^{(i)}=n_i\!-\!1\big)
\ge \lim\limits_{i\to\infty}\big(\frac{1}{4}\big)^{\frac{1}{\mu(n_i)}}
=1.
$
\end{proof}

As a consequence, the double $\la$-twisted codes of ratio $\al'/\al$
 over finite fields are asymptotically good.

\begin{theorem}\label{c good sequence over F}
Keep the notation in Remark \ref{F-notation}.
Assume that $\delta\in(0,1-q^{-1})$ satisfying that
$h_q(\frac{\al'\!+\al}{2\al''}\delta)<\frac{1}{2}$.
 Then there is a sequence $C_1,C_2,\cdots$
of double $\la$-twisted codes $C_i$
of ratio $\al'/\al$ over $F$
such that the length of $C_i$ goes to infinity,
$\lim\limits_{i\to\infty}{\rm R}(C_i)=\frac{1}{\al'\!+\al}$,
and $\Delta(C_i)>\delta$ for all $i=1,2,\cdots$.
\end{theorem}

\begin{proof} In Theorem \ref{good sequence over F}, we can take
$C_i=C_{a'\!,a}^{(i)}$ for $i=1,2,\cdots$ such that:
\begin{itemize}
\item \vskip-6pt
the length of $C_i$ is $\al' n_i+\al n_i$;
\item \vskip-4pt
the relative minimum distance $\Delta(C_i)>\delta$;
\item \vskip-4pt
the information length of $C_i$ is $\dim C_i=n_i-1$, hence the rate
\par
 ${\rm R}(C_i) =\frac{n_i-1}{\al' n_i+\al n_i}
=\frac{1}{\al'+\al}-\frac{1}{\al' n_i+\al n_i}$.
\end{itemize}
\vskip-4pt
Thus the theorem holds.
\end{proof}

\section{$(R',R)$-linear constacyclic codes of ratio $\al'/\al$}
\label{R'R c-cyclic over rings}

\begin{remark}\label{R-notation}\rm
In this section we turn to the general case and
take the notation in Definition~\ref{R'R-}:
\begin{itemize}
\item \vskip-5pt
 $R$, $R'$ are finite chain rings, $J(R)=R\pi$ of nilpotency index $\ell$,
$J(R')=R'\pi'$ of nilpotency index $\ell'$;
there is an epimorphism $\rho:R\to R'$,
hence they have the same residue field
$F:=\ol R=\ol{R'}$, and $\ell'\le\ell$, set $|F|=q$.
\item  \vskip-5pt
$\la\in R^\times$, ${\rm ord}_{F^\times}(\ol\la)=t$;
and $\la'=\rho(\la)\in R'^\times$, hence $\ol\la=\ol{\la'}$.
\item  \vskip-5pt
Integers $\al',\al>0$, $\al'\equiv\al~({\rm mod}~t)$,  $\gcd(\al,t)=1$;
further, $\al''=\min\{\al',\al\}$.
\end{itemize}
 \vskip-5pt
And further,
\begin{itemize}
\item \vskip-5pt
$n$ is a positive integer such that $\gcd(n,qt)=1$.
\item \vskip-5pt
$\delta\in(0,1-q^{-1})$.
\end{itemize}
 \vskip-5pt
Any $R[X]$-submodule $C$ of  
$\big( R'[X]\big/\langle X^{\al'\! n}-\la'\rangle \big)
  \times \big( R[X]\big/\langle X^{\al n}-\la\rangle \big)$
is an $(R',R)$-linear $(\la'\!,\la)$-constacyclic code of ratio $\al'/\al$,
code length equals $\al' n\ell'+\al n\ell$,
the relative minimum distance
$\Delta(C)=\frac{{\rm w}(C)}{\al' n\ell'+\al n\ell}$,
and the rate
${\rm R}(C)=\frac{\log_q|C|}{\al' n\ell'+\al n\ell}$.
\end{remark}

The finite chain ring $R$ has a unique minimal ideal $R\pi^{\ell-1}$, and
the following is an $R$-module epimorphism:
\begin{align*}
 R~\longrightarrow~ R\pi^{\ell-1},~~~ a~\longmapsto~ a\pi^{\ell-1}.
\end{align*}
The kernel of this $R$-module epimorphism is $R\pi$. Thus
it induces an $R$-module isomorphism
\begin{align}\label{eta}
\eta:~ F=R/R\pi~\mathop{\longrightarrow}^{\cong} R\pi^{\ell-1}, ~~~
\ol a~\longmapsto~ a\pi^{\ell-1}.
\end{align}

\begin{lemma} The following is a well-defined $R[X]$-module monomorphism
 (i.e., injective homomorphism):
\begin{equation}\label{into R[X]}\textstyle
\begin{array}{cccc}
\eta_{\al}: &F[X]\big/\langle X^{\al n}\!-\!\ol{\la}\rangle
   &{\longrightarrow}&
  R[X]\big/\langle X^{\al n}\!-\!\la\rangle, \\[4pt]
&\sum_{j=0}^{\al n-1}\ol a_jX^j & \longmapsto &
   \sum_{j=0}^{\al n-1}a_j\pi^{\ell-1}X^j,
\end{array}
\end{equation}
which preserves Hamming weights.
\end{lemma}

\begin{proof}
The above $R$-module isomorphism Eq.\eqref{eta} induces an $R[X]$-module
monomorphism, denoted by $\eta$ again:
\begin{equation}\label{F[X] to}
\textstyle
\eta:~  F[X]~\rightarrow~ R[X],~~~
 \sum_{i}\ol a_iX^i~\mapsto~ \sum_{i}a_i\pi^{\ell-1}X^i.
\end{equation}
The $R[X]$-module structure of $F[X]$ is as follows:
$g(X)\cdot f(X)=\ol g(X) f(X)$ for $f(X)\in F[X]$ and $g(X)=\sum_i b_iX^i\in R[X]$,
where $\ol g(X)=\sum_{i}\ol b_iX^i\in F[X]$.
Because $\eta$ in Eq.\eqref{F[X] to} is an $R[X]$-module monomorphism, we have:
\begin{equation}\label{R[X]-homo}
\eta\big(\ol g(X) f(X)\big)=g(X)\eta\big(f(X)\big),
 \quad \forall~ f(X)\in F[X], ~g(X)\in R[X].
\end{equation}
For $f(X)\in F[X]$ we denote $\eta(f(X))=:f^\eta(X)$.
Combining the $\eta$ in Eq.\eqref{F[X] to} with the quotient homomorphism
$R[X]\to R[X]/\langle X^{\al n}-\la\rangle$,
we obtain the following $R[X]$-homomorphism
$$
\textstyle
\tilde\eta:~  F[X]~\rightarrow~ R[X]/\langle X^{\al n}-\la\rangle,~~~
 f(X)
~\mapsto~ f^\eta(X)~({\rm mod}~X^{\al n}-\la).
$$
Assume that $f(X)=\sum_{i}\ol a_iX^i\in{\rm Ker}(\tilde\eta)$, i.e.,
there is a $g(X)=\sum_i b_iX^i\in R[X]$ such that
$$\textstyle
\eta\big(f(X)\big)=\sum_{i} a_i\pi^{\ell-1}X^i = (X^{\al n}-\la) g(X).
$$
Since $\pi\sum_{i} a_i\pi^{\ell-1}X^i=0$, we have that
$(X^{\al n}-\la) \pi g(X)=0$; further, since $X^{\al n}-\la$ is monic,
we can see that $\pi g(X)=0$. Thus, for any coefficient $b_i$ of $g(X)$
 there is a $d_i\in R$ such that $b_i=d_i\pi^{\ell-1}$.
Then $g(X)=\eta \big( d(X) \big)$
where $d(X)=\sum_i\ol d_i X^i\in F[X]$.
By Eq.\eqref{R[X]-homo},
$$
\eta\big(f(X)\big)=(X^{\al n}-\la)\eta\big(d(X)\big) =\eta\big((X^{\al n}-\ol\la)d(X)\big).
$$
Since $\eta$ is injective, $f(X)=(X^{\al n}-\ol\la)d(X)\in F[X](X^{\al n}-\ol\la)$.
We get that ${\rm Ker}(\tilde\eta)\subseteq F[X](X^{\al n}-\ol \la)$.
The inverse inclusion is obvious. Thus
$$
 {\rm Ker}(\tilde\eta)= F[X](X^{\al n}-\ol \la)=\langle X^{\al n}-\ol \la\rangle,
$$
and $\tilde\eta$ induces the $R[X]$-module monomorphism
$\eta_\al$ in Eq.\eqref{into R[X]}.

For $f(X)=\sum_{j=0}^{\al n-1}\ol a_jX^j\in F[X]/\langle X^{\al n}-\ol\la\rangle$.
the image
$$\textstyle
 \eta_\al\big( f(X)\big)=\sum_{j=0}^{\al n-1} a_j\pi^{\ell-1}X^j
  \in R[X]/\langle X^{\al n} -\la\rangle.
$$
Obviously,
$$
 \ol a_j\ne 0\quad(\mbox{in $F$}) ~~\iff~~
    a_j\pi^{\ell-1}\ne 0\quad(\mbox{in $R$}).
$$
Thus ${\rm w}\big(f(X)\big)={\rm w}\big(\eta_{\al}(f(X))\big)$; i.e.,
$\eta_{\al}$ preserves the Hamming weights.
\end{proof}

Similarly, we have the following $R'[X]$-module monomorphism
\begin{equation*}
\textstyle
\eta':~  F[X]~\rightarrow~ R'[X],~~~
 \sum_{i}\ol a_iX^i~\mapsto~ \sum_{i}a_i\pi'^{\ell'-1}X^i;
\end{equation*}
and the following $R'[X]$-module monomorphism:
\begin{equation}\label{into R'[X]}\textstyle
\begin{array}{cccc}
\eta'_{\al'}: &F[X]\big/\langle X^{\al' n}\!-\!\ol{\la'}\rangle
   &{\longrightarrow}&
  R'[X]\big/\langle X^{\al' n}\!-\!\la'\rangle, \\[4pt]
&\sum_{j=0}^{\al n-1}\ol a_jX^j & \longmapsto &
   \sum_{j=0}^{\al n-1}a_j\pi'^{\ell'-1}X^j,
\end{array}
\end{equation}
which preserves Hamming weights.
Note that, through the epimorphism $\rho:R\to R'$,
both $\eta'$ and $\eta'_\al$ can be viewed as
$R[X]$-module homomorphisms.

Recall the notation in Section \ref{double twisted over fields}
(but with assumptions in Remark \ref{R-notation}
 of this section, e.g., $\ol\la=\ol{\la'}\in F$):
\begin{itemize}
\item
$\I_{\ol\la,\al}\subseteq\R_{\ol{\la},\al}
 =F[X]/\langle X^{\al n}-\ol{\la} \rangle$,~~(Eq.\eqref{I,I_la...})
\item
$\I_{\ol{\la'},\al'}\subseteq\R_{\ol{\la'},\al'}
 =F[X]/\langle X^{\al'\! n}-\ol{\la'}\rangle$,~~(Eq.\eqref{I,I_la...})
\item
$C_{a'\!,a}=F[X](a'(X),a(X))\subseteq \I_{\ol{\la'},\al'}\times\I_{\ol\la,\al}$,~~
(Eq.\eqref{C_a',a=ga})
\item
positive integers $n_1,n_2,\cdots$ satisfy Eq.\eqref{n_i...}
\item $C_{a'\!,a}^{(1)},~C_{a'\!,1}^{(2)},~\cdots$,~ in Eq.\eqref{sequence over F}.
\end{itemize}

By the $R[X]$-monomorphisms
Eq.\eqref{into R[X]} and Eq.\eqref{into R'[X]},
we can embed
$$C_{a'\!,a}^{(i)}=
 \big\{\,(g(X)a'(X),\,g(X)a(X))\;\big|\;g(X)\in F[X]\big\}\qquad
(\mbox{see Eq.\eqref{C_a',a=ga})} $$
 into
$$\big( R'[X]/\langle X^{\al' n_i}-\la'\rangle \big) \times
  \big( R[X]/\langle X^{\al n_i}-\la\rangle \big)$$
as follows:
\begin{align}
\big(g(X)a'(X),\,g(X)a(X)\big)~\longmapsto~
\big(\eta'(g(X))\eta'_{\al'}(a'(X)),\,\eta(g(X))\eta_\al(a(X))\big).
\end{align}
We denote the image of $C^{(i)}_{a'\!,a}$ by $\tilde C^{(i)}_{a'\!,a}$.
In this way, we obtain a sequence of
$(R'\!,R)$-linear $(\la'\!,\la)$-constacyclic codes of ratio
$\al'/\al$ as follows:
\begin{align}\label{sequence over R}
\tilde C_{a'\!,a}^{(1)},~ \tilde C_{a'\!,a}^{(2)},~\cdots,~\tilde C_{a'\!,a}^{(i)},~\cdots
\end{align}

\begin{theorem}\label{good sequence over R}
Let notation be as in Eq.\eqref{sequence over R}.
Assume that $\delta\in(0,1-q^{-1})$ satisfying that
$h_q(\frac{\al'\!+\al}{2\al''}\delta)<\frac{1}{2}$.
Then

{\bf(1)} $\lim\limits_{i\to\infty}\Pr\big(
  \Delta(\tilde C_{a'\!,a}^{(i)})
  >\frac{\al' +\al }{\al'\ell'+\al \ell}\!\cdot\!\delta
  \big)=1$.

{\bf(2)} $\lim\limits_{i\to\infty}
 \Pr\big(|\tilde C_{a'\!,a}^{(i)}|=q^{n_i\!-\!1}\big)=1$.
\end{theorem}

\begin{proof}
Because both Eq.\eqref{into R[X]} and Eq.\eqref{into R'[X]}
are monomorphism and preserve
the Hamming weights,
${\rm w}(\tilde C_{a'\!,a}^{(i)})={\rm w}(C_{a'\!,a}^{(i)})$.
The code length of $C_{a'\!,a}^{(i)}$ is $\al'n+\al n$; while
the code length of $\tilde C_{a'\!,a}^{(i)}$ is $\al'n\ell'+\al n\ell$. So
$$
 \Delta(\tilde C_{a'\!,a}^{(i)})
=\frac{{\rm w}\big(C_{a'\!,a}^{(i)}\big)}{\al'n\ell'+\al n\ell}
=\frac{{\rm w}\big(C_{a'\!,a}^{(i)}\big)}{\al' n+\al n}
 \cdot\frac{\al' n+\al n}{\al'n\ell'+\al n\ell}
=\frac{\al' +\al }{\al'\ell'+\al \ell}\cdot\Delta(C_{a'\!,a}^{(i)});
$$
then
$$\textstyle
\Delta(C_{a'\!,a}^{(i)})\ge\delta
 ~\iff~
 \Delta(\tilde C_{a'\!,a}^{(i)})\ge \frac{\al' +\al }{\al'\ell'+\al \ell}\cdot\delta.
$$
Therefore,
the theorem follows from Theorem
\ref{good sequence over F} immediately.
\end{proof}

Finally, the following is a more precise version of Theorem \ref{main}.

\begin{theorem}\label{c good sequence over R}
Assume that $\delta\in(0,1-q^{-1})$ satisfying that
$h_q(\frac{\al'\!+\al}{2\al''}\delta)<\frac{1}{2}$.
 Then there is a sequence $C_1,C_2,\cdots$
of $(R'\!,R)$-linear $(\la'\!,\la)$-constacyclic codes $C_i$
of ratio $\al'/\al$
such that the length of $C_i$ goes to infinity,
$\lim\limits_{i\to\infty}{\rm R}(C_i)=\frac{1}{\al'\ell'\!+\al\ell}$,
and $\Delta(C_i)>\frac{\al' +\al }{\al'\ell'+\al \ell}\!\cdot\!\delta$ for all $i=1,2,\cdots$.
\end{theorem}

\begin{proof} In Theorem \ref{good sequence over R}, we can take
$C_i=\tilde C_{a'\!,a}^{(i)}$ for $i=1,2,\cdots$ such that:
\begin{itemize}
\item \vskip-6pt
the length of $C_i$ is $\al' n_i\ell'+\al n_i\ell$;
\item \vskip-4pt
the relative minimum distance
$\Delta(C_i)>\frac{\al' +\al }{\al'\ell'+\al \ell}\cdot\delta$;
\item \vskip-4pt
the information length of $C_i$ is $\log_q|C_i|=n_i-1$, hence the rate
\par
 ${\rm R}(C_i) =\frac{n_i-1}{\al' n_i\ell'+\al n_i\ell}
=\frac{1}{\al'\ell'+\al\ell}-\frac{1}{\al' n_i\ell'+\al n_i\ell}$.
\end{itemize}
\vskip-4pt
Thus the theorem holds.
\end{proof}

\section{Conclusion}\label{conclusion}

The main contribution of this paper is that   
a very general type of codes is constructed and 
the asymptotic goodness of such codes is proved.

We introduced a type of codes:  
let $R$ and $R'$ be two finite commutative chain rings 
with an epimorphism $\rho: R\to R'$, let $\lambda\in R^\times$
and $\lambda'=\rho(\lambda)$, and $\alpha,\alpha',n$ be positive integers;
we call any $R[X]$-submodule $C$ of the $R[X]$-module 
$\big( R'[X]/\langle X^{\alpha'\! n}-\lambda'\rangle\big)
\times \big( R[X]/\langle X^{\alpha n}-\lambda\rangle\big)$
by an $(R', R)$-linear constacyclic code.
Such codes form an extensive class of codes.  
First, the two alphabets for the codes are finite commutative chain rings
which cover many alphabets used in coding.
Second, instead of cyclic structures, 
the more general constacyclic structures are considered.
Third, the two lengths of the two constacyclic circles are not necessarily equal.
Thus, $(R', R)$-linear constacyclic codes cover many well-known kinds of codes,
e.g., quasi-cyclic codes of fractional index, 
${\Bbb Z}_2{\Bbb Z}_4$-additive cyclic codes etc.

We proved in a random style that 
$(R', R)$-linear constacyclic codes are asymptotically good.
The usual probabilistic method applied to quasi-cyclic codes
could not applied to the constacyclic case directly. 
We take an algebraic skill to reform it into a developed probabilistic method
effective for studying quasi-constacyclic codes.
And then we reduced the proof for the asymptotic goodness 
of $(R', R)$-linear constacyclic codes to the quasi-constacyclic case, 
so that the proof of the asymptotic goodness of
$(R', R)$-linear constacyclic codes is completed. 
The developed probabilistic method is another contribution of this paper.


\end{document}